\documentclass[a4paper, 11pt]{article}

\usepackage{amsmath}
\usepackage{amsthm}
\usepackage{amssymb}
\usepackage{ifpdf}
\usepackage{graphicx}
\usepackage{xcolor}
\usepackage{algorithm}
\usepackage[]{algpseudocode}
\usepackage{subfigure}

\usepackage{geometry}
\newtheorem{theorem}{Theorem}
\newtheorem{lemma}{Lemma}

\newcommand{\Real}{\mathbb{R}}
\newcommand{\Natural}{\mathbb{N}}
\newcommand{\Integer}{\mathbb{Z}}

\newcommand{\St}{\mathcal{S}}
\newcommand{\state}[2]{\langle #1, #2 \rangle}
\newcommand{\stateY}[2]{\state{#1}{#2}^\mathcal{Y}}
\newcommand{\stateK}[2]{\state{#1}{#2}^\mathcal{K}}
\newcommand{\stateV}[2]{\state{#1}{#2}^\mathcal{V}}
\newcommand{\seg}[1]{\mathcal{H}_{#1}}
\newcommand{\s}{\tilde{s}}
\newcommand{\dil}{\tilde{\delta}}

\DeclareMathOperator{\jump}{jump}
\DeclareMathOperator{\dist}{dist}
\DeclareMathOperator{\cl}{cl}

\newcommand{\fig}[3]{
    \begin{figure}[ht]
        \centering
        \includegraphics{#3}
        \caption{#2}
        \label{#1}
    \end{figure}
}

\newif\ifnotesw\noteswtrue 

\begin{document}
\title{Linear Time Algorithm for Optimal Feed-link Placement}

\date{}
\author{
    {Marko Savi\'{c}
        \thanks{Department of Mathematics and Informatics, University of Novi Sad, Serbia. Email: marko.savic@dmi.uns.ac.rs.
        Partly supported by Ministry of Education and Science, Republic of Serbia}
    }
    \qquad
    {Milo\v{s} Stojakovi\'{c}
        \thanks{Department of Mathematics and Informatics, University of Novi Sad, Serbia. Email: milos.stojakovic@dmi.uns.ac.rs.
        Partly supported by Ministry of Education and Science, Republic of Serbia, and Provincial Secretariat for Science, Province of Vojvodina.}
    }
}

\maketitle

\begin{abstract}
Given a polygon representing a transportation network together with a point $p$ in its interior, we aim to extend
the network by inserting a line segment, called a feed-link, which connects $p$ to the boundary of the polygon.
Once a feed link is fixed, the geometric dilation of some point $q$ on the boundary is the ratio between the length of the shortest path from $p$ to
$q$ through the extended network, and their Euclidean distance. The utility of a feed-link is inversely proportional to
the maximal dilation over all boundary points.

We give a linear time algorithm for computing the feed-link with the
minimum overall dilation, thus improving upon the previously known algorithm of complexity that is roughly $O(n \log n)$.
\end{abstract}

\section{Introduction}

Depending on the requirements, there are many standard ways to connect a new node to an existing network. Probably the
most straightforward one is to simply snap the location to the closest point on the network. This may be unsuitable as
the node location is modified. Also, it can happen that two points geometrically close to each other are snapped to
parts of network that are far away, which may be undesirable. Another approach is to link all new nodes inside a
network face to the \emph{feed-node} which is then connected to the network. Alternatively, each new node can be
individually attached to the network using a \emph{feed-link}. This approach was taken e.g.\
in~\cite{dahlgren2006evaluation, dahlgren2007development}, where the new location is simply connected to the nearest
existing location by a feed-link.

In an attempt to reduce unnecessary detours, Aronov et al.\ in~\cite{aronov2012connect} introduce a more sophisticated
way of choosing where on the existing network to attach the new feed-link, using the so-called \emph{dilation} to
measure the quality of feed-links. For a planar embedding of a graph and two different points $p$ and $q$ on it, we
define the \emph{detour} (sometimes called slightly less formally the \emph{crow flight conversion coefficient}) as the
ratio of the minimum distance of points $p$ and $q$ staying on the graph, and their Euclidean (crow flight) distance.
The \emph{geometric dilation} of the graph is the maximum detour taken over all pairs of points on the graph. For a
detailed view on geometric dilation and related concepts, we refer the reader to~\cite{ebbers2006geometric}.

Following the approach presented in~\cite{aronov2012connect}, our goal is to attach a new node positioned inside a face
of an existing network in a suitable way. The face is given as the boundary of a polygon $P$ and the new node is
a point $p$ inside the polygon. We want to connect the point $p$ to a point on the polygon boundary $P$ using a single
line segment. Such connection is called a feed-link. Note that a feed-link may have more than one point in the
intersection with the polygon boundary, but we do not regard these points as connection points. An optimal feed-link is
the one that minimizes the maximum detour ratio from point $p$ to a point on the boundary.

The problem of finding the optimal feed-link is analyzed in~\cite{aronov2012connect} and an algorithm that runs in
$O(\lambda_7(n) \log n)$ time is presented, where $n$ is the number of vertices on the boundary of the face, and
$\lambda_7(n)$ is the maximum length of a Davenport-Schinzel sequence of order 7 on $n$ symbols, a slightly superlinear
function. Here, we make an improvement by presenting an $O(n)$ time algorithm that finds an optimal feed-link.

Although the initial problem statement assumes that $p$ lies inside the polygon and the polygon is simple, all of our calculations work out
exactly the same for an arbitrary point $p$ in the plane and arbitrary polygons, possibly self-intersecting, and we obtain the same result in that more general setting.

The rest of the paper is organized as follows. In Section~\ref{sec:NotationAndProblemStatement} we introduce notation and a formal definition of the problem. An alternative view that will help us with the analysis is presented in Section~\ref{sec:AnotherViewOfTheProblem}. Next, in Section~\ref{sec:SlidingLeverAlgorithm} we describe an algorithm that outputs a discretized descriptions of the components and finally, in Section~\ref{sec:MergingTheTwoDilations}, we show how to combine those outputs to give the solution of the original problem.

\section{Notation and problem statement}
\label{sec:NotationAndProblemStatement}

A polygon, which is not necessarily simple, is given as the list of its vertices $v_0, v_1, \ldots, v_{n-1}$ in the plane. By $P$ we denote only the boundary of that polygon. We are also given the point $p$ lying in the same plane as $P$. A \emph{feed-link} is a line segment $pq$, connecting $p$ with some point $q \in P$.

For any two points $q, r \in P$ \emph{dilation of $r$ via $q$} is defined as
\[ \delta_q(r) = \frac{|pq| + \dist(q,r)}{|pr|} ,\]
where $\dist(q,r)$ is the length of the shortest route  between $q$ and $r$ over the polygon's boundary, and $|ab|$ is the Euclidean distance between points $a$ and $b$, see Figure~\ref{fig:1}.
\fig{fig:1}{The concept of dilation.}{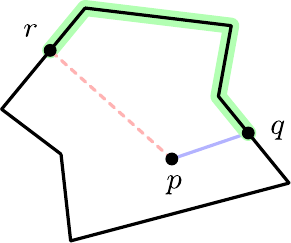}

For a point $q \in P$, \emph{dilation via $q$} is defined as
\[ \dil_q = \max_{r \in P} \delta_q(r) .\]

The problem of finding the optimal feed-link is to find $q$ such that $\dil_q$ is minimized.

\subsection{Left and right dilation}

Given two points $a, b \in P$, $P[a,b]$ is the portion of $P$ obtained by going from $a$ to $b$ around the polygon in the positive direction, including points $a$ and $b$. Let $\mu(a,b)$ be the length of $P[a,b]$, and $\mu(P)$  the perimeter of $P$.

For given $a \in P$, let $a'$ be the point on $P$, different from $a$, for which $\mu(a,a') = \mu(a',a) = \mu(P)/2$, see Figure~\ref{fig:2}. By $P^+[a]$ we denote $P[a,a']$, and by $P^-[a]$ we denote $P[a',a]$. Obviously, $P^+[a] \cup P^-[a] = P$ and $P^+[a] \cap P^-[a] = \{a,a'\}$.
\fig{fig:2}{Left and right portion of $P$ observed from point $a$.}{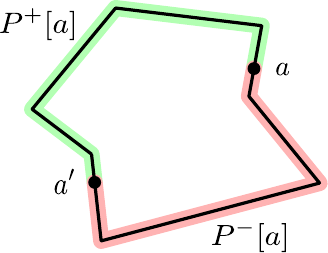}

Given point $q \in P$ and $r \in P^+[q]$, the \emph{left dilation of $r$ via $q$} is defined as
\[ \delta^+_q(r) = \frac{|pq| + \mu(q,r)}{|pr|} .\]
On the other hand, for $r \in P^-[q]$, the \emph{right dilation of $r$ via $q$} is defined as
\[ \delta^-_q(r) = \frac{|pq| + \mu(r,q)}{|pr|} .\]

When measuring $\dist(q,r)$, the shortest path from $q$ to $r$ over $P$ must lie entirely either in $P^+[q]$ or $P^-[q]$. This allows us to express the dilation of $r$ via $q$ using left and right dilations of $r$ via $q$
\[
    \delta_q(r) =
        \begin{cases}
             \delta^+_q(r), & \mbox{if $r \in P^+[q]$} \\
             \delta^-_q(r), & \mbox{if $r \in P^-[q]$} \\
        \end{cases}.
\]

Given point $q \in P$, the \emph{left dilation via $q$} is defined as $\dil^+_q = \max_{r \in P^+[q]} \delta^+_q(r)$, and the \emph{right dilation via $q$} as $\dil^-_q = \max_{r \in P^-[q]} \delta^-_q(r)$. Finally, the dilation via $q$ can be expressed as
\begin{equation}
\label{eq:dilationIsMaximum}
\dil_q = \max(\dil^+_q, \dil^-_q) = \max_{r \in P}\delta_q(r).
\end{equation}

In the following two sections we will be concerned only with the left
dilation; as the problem of finding the right dilation will turn out to be analogous. In
Section~\ref{sec:MergingTheTwoDilations} we will show how to combine our findings
about the left and right dilation to provide the answer to the original
question.
To simplify the notation, we will not use the superscript $+$ in Sections~\ref{sec:AnotherViewOfTheProblem} and~\ref{sec:SlidingLeverAlgorithm} assuming that we deal with the left dilation.

\section{Another view of the problem}
\label{sec:AnotherViewOfTheProblem}

We parametrize points on $P$ by defining $P(t)$, $t \in \Real$, to be the point on $P$ for which $\mu(v_0,P(t)) \equiv t \pmod{\mu(P)}$, see Figure~\ref{fig:3}.
\fig{fig:3}{Parametrization of $P$.}{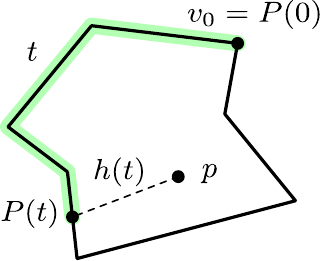}

The distance of a point to points on a straight line is known to be
a hyperbolic function. The plot of the distance function $h(t) :=
|pP(t)|$ is an infinite sequence of hyperbola segments joined at
their endpoints, where $(kn+r)$-th hyperbola segment corresponds to
the $r$-th side of $P$, for $r \in \{0, 1, \ldots, n-1\}$ and $k \in
\Integer$, see Figure~\ref{fig:4}. For each $i = kn+r$, hyperbola
$h_i$ is of the form $h_i(t) = \sqrt{(t-m_i)^2+d_i^2}$, for some
values $m_i$ and $d_i$, so that $m_{kn+r} = m_r + k\mu(P)$,
$d_{kn+r} = d_r$, and $d_r$ is the distance between $p$ and the line
containing the $r$-th side of $P$. The left endpoint of $i$-th
hyperbola segment is $E_i := (e_i, h(e_i))$, where $e_i = k\mu(P) +
\mu(v_r)$, and the right endpoint is at $(e_{i+1}, h(e_{i+1}))$. We
will consider that each hyperbola segment contains its left
endpoint, but not the right endpoint. By $H(t)$ we denote the point
on the plot of $h$ corresponding to the parameter $t$, so $H(t) :=
(t, h(t))$. The plot is, obviously, periodic, with the period of
$\mu(P)$, that is, $H(t) = H(t + k\mu(P))$. We denote the $i$-th
hyperbola segment with $\seg{i}$. \fig{fig:4}{The plot of
$h(t)$.}{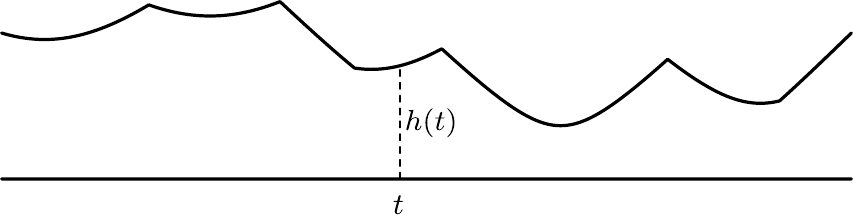}

Let $o(t) = t - h(t)$, and $O(t) = (o(t), 0)$, see
Figure~\ref{fig:5}. We also define $o_i(t) = t - h_i(t)$, and
$O_i(t) = (o_i(t), 0)$. Given points $q\in P$ and $r\in P^+[q]$, we
have their corresponding parameters $t_q$ and $t_r$, such that $t_q
\leq t_r \leq t_q + \mu(P)/2$. The slope of the line passing through
points $O(t_q)$ and $H(t_r)$ is
\begin{equation}
\begin{aligned}
s(t_q, t_r) &:= \text{slope}\left(\ell(O(t_q),H(t_r))\right) \\
            &= \frac{h(t_r)}{t_r - t_q + h(t_q)} \\
            &= \frac{|pP(t_r)|}{\mu(P(t_q),P(t_r)) + |pP(t_q)|} \\
            &= \frac{1}{\delta^+_{P(t_q)}(P(t_r))}, \\
\end{aligned}
\end{equation}
hence the slope between
$O(t_q)$ and $H(t_r)$ is equal to the inverse of the left dilation
of $r$ via $q$.

\fig{fig:5}{Dilation and slope relation.}{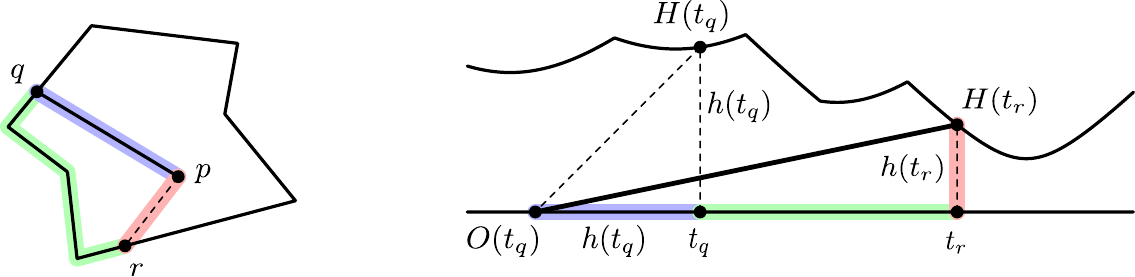}

We define $\s(t_q)$ to be the lowest slope from $O(t_q)$ to $H(t_r)$ among all $t_r \in [t_q, t_q + \mu(P)/2]$. From the previous observation it follows that this slope equals the inverse of the left dilation via
$q$,
\begin{equation}
\label{eq:dilationSlopeConnection}
\begin{aligned}
\s(t_q) &:= \min_{t_r \in [t_q, t_q + \mu(P)/2]} s(t_q, t_r) \\
        &= \min_{t_r \in [t_q, t_q + \mu(P)/2]} \frac{1}{\delta^+_{q}(r)} \\
        &= \frac{1}{\max_{r \in P^+[q]} \delta^+_{q}(r)} \\
        &= \frac{1}{\dil^+_{q}}.
\end{aligned}
\end{equation}

Obviously, $\s(t) \in (0, 1]$ because it is strictly positive and $\s(t) \leq s(t, t) = 1$. 
This enables us to estimate dilation by looking at the slope of the line we just defined.

\begin{lemma}
    \label{lem:plotSlope}
    For any two distinct values of $t_1$ and $t_2$,
    \[|(h(t_2) - h(t_1)) / (t_2 - t_1)| \leq 1. \]
\end{lemma}
\begin{proof}
Function $h(t)$ is continuous and, as a union of countably many
segments of hyperbolas with first derivatives less or equal to one
in absolute value, is differentiable almost everywhere having
$|h'(t)| < 1$ for each $t \in \Real \setminus \{e_i : i \in \Natural
\}$. This property readily implies the statement of the lemma.
\end{proof}

So far, $\s(t_q)$ was defined as minimum only among slopes $s(t_q,
t_r)$ where $t_r$ belongs to the interval $[t_q, t_q + \mu(P)/2]$. However, from
Lemma~\ref{lem:plotSlope} follows that $s(t_q, t_r)$ cannot be less
than $1$ when $t_r \in [o(t_q), t_q]$, and $\s(t_q)$ is at most $1$, so this interval can be extended, and we have
\[ \s(t_q) = \min_{t_r \in [o(t_q), t_q + \mu(P)/2]} s(t_q, t_r) .\]

\section{Sliding lever algorithm}
\label{sec:SlidingLeverAlgorithm}

\subsection{Lever}

For a fixed $t$, consider the line segment having slope $\s(t)$, with
one endpoint at $O(t)$ and the other at $(t+\mu(P)/2,
\s(t)(\mu(P)/2+h(t)))$, see Figure~\ref{fig:6}. Let us call that line
segment the \emph{lever for $t$}. Note that the lever only touches
the plot, never intersecting it properly.

Let $C(t)$ be the leftmost point in which the lever for $t$ touches
the plot, and let $c(t)$ be such that $H(c(t))=C(t)$. Then $c(t) \in
[o(t), t+\mu(P)/2]$ and it is the lowest value in this interval for
which $\s(t) = s(t, c(t))$. Coming back to the original setup, this
means that left dilation via $P(t)$ reaches its maximum for
$P(c(t))$. \fig{fig:6}{Lever.}{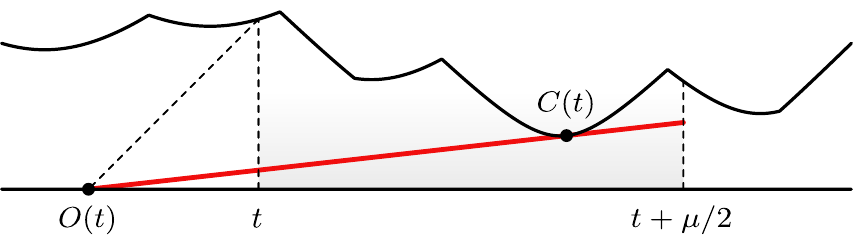}

We now continuously decrease parameter $t$ and observe what is
happening with the updated lever. The following monotonicity lemma
states that when $t$ is decreasing $o(t)$ and $c(t)$ are decreasing
as well, which means that decreasing $t$ corresponds to ``dragging"
the lever in the leftward direction.

\begin{lemma}
\label{lem:monotonicity} For $t_1 < t_2$ we have $o(t_1) \leq o(t_2)$ and $c(t_1) \leq c(t_2)$.
\end{lemma}
\begin{proof}
Suppose $t_1 < t_2$. Using Lemma~\ref{lem:plotSlope} we get
\begin{align*}
(h(t_2) - h(t_1))/(t_2 - t_1) &\leq 1, \\
t_1 - h(t_1) &\leq t_2 - h(t_2), \\
o(t_1) &\leq o(t_2). \\
\end{align*}
To show that $c(t_1) \leq c(t_2)$, assume the opposite, that $c(t_1)
> c(t_2)$. Then, $t_1 < t_2 < c(t_2) < c(t_1) \leq t_1 + \mu(P)/2$,
and $c(t_1) \in [t_2, t_2 + \mu(P)/2]$.

Suppose first that the line segments $O(t_1)C(t_1)$ and
$O(t_2)C(t_2)$ do not intersect. Since $o(t_1) \leq o(t_2)$, the
point $O(t_2)$ is not above $O(t_1)C(t_1)$, so the segment
$O(t_2)C(t_2)$ lies completely under $O(t_1)C(t_1)$. Since
$O(t_2)C(t_2)$ touches the plot, the plot must intersect
$O(t_1)C(t_1)$ in some point left of $t_1$, which is a contradiction
since $C(t_1)$ is the leftmost point where the lever for $t_1$
touches the plot.

Since the line segments $O(t_1)C(t_1)$ and $O(t_2)C(t_2)$ do intersect, the point $C(t_1)$ lies under the segment $O(t_2)C(t_2)$. Thus we have
\[
s(t_2, c(t_1)) < s(t_2, c(t_2))
               = \min_{t \in [t_2, t_2 + \mu(P)/2]} s(t_2, t)
               \leq s(t_2, c(t_1)),
\]
which is a contradiction. This concludes the proof of the lemma.
\end{proof}

\subsection{States}

In order to be able to simulate the continuous leftward motion of the lever, we transform it to an iteration over a discrete sequence of states. We define different lever states depending on how the lever is positioned relative to the sequence of hyperbola segments.

When $t \in [e_i, e_{i+1})$ and $c(t) \in [e_j, e_{j+1})$, we say
that the lever for $t$ is in \emph{phase} $\state{i}{j}$. Phase in which is the lever, together with the manner in which the lever touches the plot define the \emph{state of the lever}. There are three possible ways for the lever to touch the plot, denoted by $\mathcal{K}$ (arc tangency), $\mathcal{Y}$ (endpoint sliding), and $\mathcal{V}$ (wedge touching).

\begin{itemize}
\item State $\stateK{i}{j}$ : $c(t) < t + \mu(P)/2$ and the lever is the tangent to $\seg{j}$.

When $t$ is decreasing, the lever is sliding to the left along the
$h_j$ maintaining the tangency, thus continuously decreasing the
slope.

\fig{fig:7k}{State $\stateK{i}{j}$}{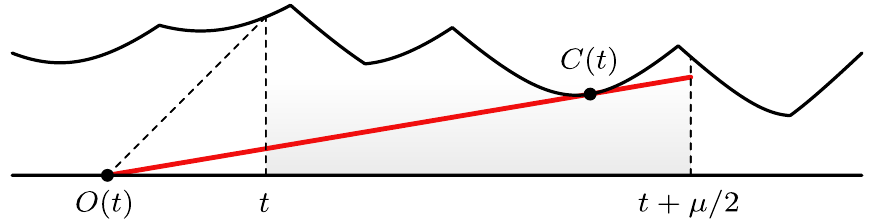}

\item State $\stateY{i}{j}$ : $c(t) = t + \mu(P)/2$.

Point $C(t)$ is the right endpoint of the lever. It is the only
point where lever touches the plot. When $t$ is decreasing, the
lever is moving to the left while keeping its right endpoint on
$h_j$.

\fig{fig:7y}{State $\stateY{i}{j}$}{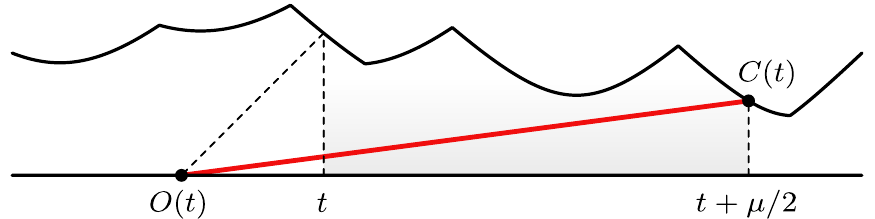}

\item State $\stateV{i}{j}$ : $c(t) < t + \mu(P)/2$ and the lever is passing through the point $H(e_j)$, the endpoint between hyperbola segments $\seg{j-1}$ and $\seg{j}$.

This situation occurs only if $m_{j-1} > m_j$. The two neighboring
hyperbola segments then form a ``wedge" pointing downwards, and when
$t$ is decreasing the lever is sliding to the left while maintaining
the contact with the tip of that wedge, thus continuously decreasing
the slope.

\fig{fig:7v}{State $\stateV{i}{j}$}{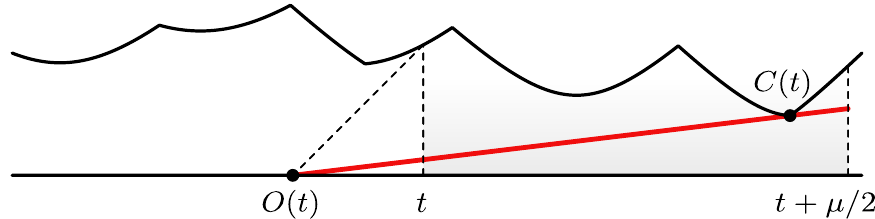}

\end{itemize}

\subsection{Events}

In the process of decreasing $t$ and dragging the left endpoint of the lever towards left, the lever state changes at certain moments. We call such events \emph{state transition events}. It is crucial for us to be able to efficiently calculate where those events can occur. If current lever position, $t_c$, and the current state are known, the following event can be determined by maintaining the set of conceivable future events of which at least one must be realized, and proceeding to the one that is first to happen, i.e.\ the one with the largest $t$ not larger than $t_c$. To do that, we must know how to calculate the value of $t$ for each of those events.

\subsubsection{Jumping and retargeting}

We will first devote some attention to the most challenging kind of events, which we call \emph{jumping events}. These are events in which $C(t)$ abruptly changes its position by switching to a different hyperbola segment. In order to efficiently find state transition events that include jumps, we always need to know to which hyperbola segment we can jump to from the current position. There is always at most one such target hyperbola segment, and we will show how to keep track of it.

Consider some point $H(x)$ on the plot. Let $\jump(x)$, \emph{the jump destination for $x$}, be the index of a hyperbola segment which contains the rightmost point $H(w)$ on the plot such that $w < x$ and the ray from $H(x)$ through $H(w)$ only touches the plot, but does not intersect it properly. That is, $\jump(x)$ is the index of the lowest visible hyperbola segment when looking from the point $H(x)$ to the left. If there is no such $w$, because hyperbola segments on the left are obscured by the segment containing $H(x)$, then we set $\jump(x)$ to be the index of the hyperbola segment containing $H(x)$.

Consider only the values of $x$ at which $\jump(x)$ changes value.
We call such values \emph{retargeting positions}, and points $H(x)$
\emph{retargeting points}, see Figure~\ref{fig:tangent}. There are two types of retargeting points. \emph{Retargeting points of the first type} are the points on $\seg{k}$ in which jump destination changes from $i$ to $j$, where $i < j < k$, (Figure~\ref{fig:tangent1}). \emph{Retargeting points of the second type} are the points on $\seg{k}$ in which jump destination changes from $k$ to $j$, where $j < k$, (Figure~\ref{fig:tangent2}).

\begin{figure}[ht]
	\centering
	\subfigure[Retargeting point of the first type.]{
		\label{fig:tangent1}
		\includegraphics{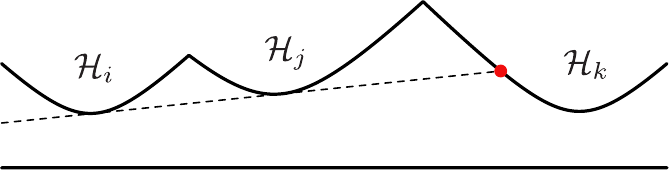}
	}
   	\hspace{10pt}
	\subfigure[Retargeting point of the second type.]{
		\label{fig:tangent2}
		\includegraphics{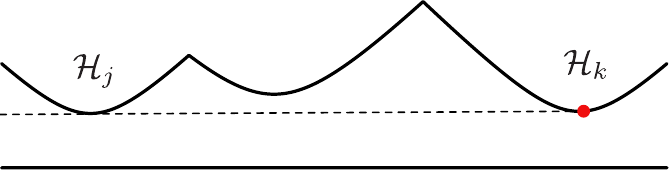}
	}
	\caption{Two types of retargeting points.}
	\label{fig:tangent}
\end{figure}

\begin{theorem}
\label{thm:RetargetingPositions}
    It is possible to find all retargeting positions, ordered from left to right, together with jump destinations of those positions, in $O(n)$ time.
\end{theorem}

\begin{proof}
    The algorithm for finding all retargeting points is similar to finding the lower convex chain in Andrew's monotone chain convex hull algorithm \cite{andrew1979another}. Our algorithm, however, runs in linear time because the sequence of hyperbola segments is already sorted. Before we give the algorithm, we describe the process and the supporting structure in more detail.

Given a set $A$ of hyperbola segments, we take a look at the convex hull of their union and divide its boundary into the upper and lower part (i.e., the part lying above the segments, and the one below the segments). We are interested only in those hyperbola segments from $A$ that have a nonempty intersection with the lower part of the convex hull boundary, having in mind that each hyperbola segment contains its left endpoint, but not the right one. Let us call the sequence of all such hyperbola segments, ordered from left to right, the \emph{convex support} for $A$, see Figure~\ref{fig:convex}.

\fig{fig:convex}{Convex support for all shown hyperbola segments is marked with solid lines.}{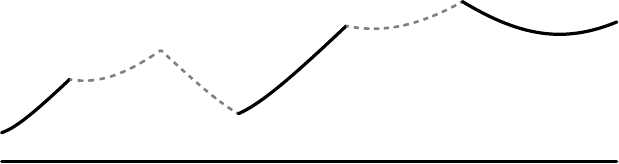}

We say that three hyperbola segments from the plot of $h$ are in convex position if no line segment connecting a point from the left and a point from the right hyperbola segment passes completely below the middle hyperbola segment. Note that any three hyperbola segments of any convex support are in convex position.

Let $j_0$ be the index of the hyperbola segment which contains any of the global minima of the plot of $h$. Starting from $\{\seg{j_0-1},\seg{j_0}\}$, we process segments from left to right and maintain the convex support for the set $\{\seg{j_0-1}, \seg{j_0}, \ldots, \seg{k}\}$, where $k$ is the index of the segment being processed.

We will use a stack to represent the convex support (only the indices of hyperbola segments are stored). Suppose the stack already contains the convex support for $\{\seg{j_0-1}, \seg{j_0}, \ldots, \seg{k-1}\}$, and we want to add a new segment $\seg{k}$. We must make changes to the stack, so that it now represents a convex support for the new, extended, set $\{\seg{j_0-1}, \seg{j_0}, \ldots, \seg{k}\}$. To achieve this, we pop segments from the stack until the last two segments still in the stack, together with $\seg{k}$, are in convex position. (Note that $\seg{j_0}$ will never be popped this way, as it contains a global minimum.) Finally, in case $\seg{k}$ belongs to the convex support of $\{\seg{j_0-1}, \seg{j_0}, \ldots, \seg{k}\}$, we push it on the stack .

Let $\cl(X)$ denote the closure of a point set $X$, so $\cl(\seg{i}) = \seg{i} \cup \{E_{i+1}\}$, and let us consider the line $l$ touching both $\cl(\seg{i})$ and $\cl(\seg{j})$, $i < j$, from below. If such a line is not unique, which can possibly happen only when $j=i+1$, we take $l$ to be the line with the smallest slope (that is, the line tangent to $\cl(\seg{i})$ in $E_{i+1}$). We call the line $l$ the \emph{common tangent} of $\seg{i}$ and $\seg{j}$. It can be computed in constant time, and in the following algorithm it is obtained as the return value of the function $\Call{Tangent}{i, j}$.

Note that if $Z$ is a point with larger first coordinate than the point $l \cap \cl(\seg{j})$ (i.e.\ $Z$ is to the right of $l \cap \cl(\seg{j})$) and below the graph, then the point $Z$ sees $\seg{i}$ lower than $\seg{j}$ if $Z$ is below $l$, and $\seg{i}$ lower than $\seg{j}$ if $Z$ is above $l$. We will use this fact in our analysis.

Now we are ready to present Algorithm~\ref{alg:RetargetingPoints} that shows how we get retargeting points as intersections of hyperbola segments and common tangents of successive segments from the convex support.

\begin{algorithm}
\caption{Retargeting Points}
\label{alg:RetargetingPoints}
\begin{algorithmic}
    \State $RetargetingPoints \leftarrow [\ ]$
    \State \Call{Push}{$j_0-1$}
    \State \Call{Push}{$j_0$}
    \For {$k \leftarrow j_0 + 1$ to $j_0 + n$}
        \Loop
            \State $i \leftarrow$ Second-to-top element of the stack
            \State $j \leftarrow$ Top element of the stack
            \State $l_1 \leftarrow \Call{Tangent}{i, j}$
            \If {$l_1 \cap \seg{k} \neq \emptyset$}
                \State Let $g$ be the leftmost point of $l_1 \cap \seg{k}$.
                \State Append $g$ to $RetargetingPoints$, and set jump destination of $g$ to $j$.
                \State $\Call{Pop}$
            \Else
	            \State $l_2 \leftarrow \Call{Tangent}{j, k}$
                \If {$l_2 \cap \seg{k} \neq \emptyset$}
                    \State Let $g$ be the only point of $l_2 \cap \seg{k}$.
                    \State Append $g$ to $RetargetingPoints$, and set jump destination of $g$ to $j$.
                    \State \Call{Push}{$k$}
                \EndIf
                \State \textbf{break loop}
            \EndIf
        \EndLoop
    \EndFor
\end{algorithmic}
\end{algorithm}

To show the correctness of this algorithm, we first observe that each reported point must be a retargeting point since the jump destination changes at it.

Indeed, points reported in the outer ``if" branch lie on the common tangent of two successive hyperbola segments $\seg{i}$ and $\seg{j}$ from the convex support, and $\seg{k}$ is the first segment to be intersected by that tangent. The point of the intersection is the boundary between the region of $\seg{k}$ from which $\seg{i}$ is the lowest segment when looking to the left and the region of $\seg{k}$ for which such lowest segment is $\seg{j}$, as shown in Figure~\ref{fig:tangent1}. Thus, a point $g$ reported in this branch has the property that points on $\seg{k}$ just left and just right of the point $g$ have $\seg{j}$ and $\seg{i}$ as their jump destinations, respectively, so it is a retargeting point of the first type.

Inner ``if" branch occurs when the segment $\seg{k}$ is appended to the convex support, in which case there is a point on $\seg{k}$ acting as a boundary between the region of $\seg{k}$ from which no other segment is visible (when looking to the left), and the region of $\seg{k}$ from which $\seg{j}$ is visible, and such point lies on the common tangent of $\seg{j}$ and $\seg{k}$, as shown in Figure~\ref{fig:tangent2}. Therefore, the point $g$ reported in this branch has the property that points on $\seg{k}$ just left and just right of the point $g$ have $\seg{j}$ and $\seg{k}$ as their jump destinations, respectively, so it is a retargeting point of the second type.

Next, let us make sure that no retargeting points were omitted by this algorithm. Consider a retargeting point $g$ lying on $\seg{k}$.

If $g$ is retargeting point of the first type, it must lie on $\Call{Tangent}{i, j}$ for some $i < j < k$ (Figure~\ref{fig:tangent1}). Note that there can be no $\seg{r}$ with $r < k$ such that it reaches below $\Call{Tangent}{i, j}$; otherwise $r$ would be a jumping destination for $g$. That implies that both $\seg{i}$ and $\seg{j}$ are the part of the lower convex hull of the segments left from $\seg{k}$, which further means that $\seg{i}$ and $\seg{j}$ are two consecutive elements of the convex support for the set $\{\seg{j_0-1}, \seg{j_0}, \ldots, \seg{k-1}\}$. Since $\Call{Tangent}{i, j}$ intersects $\seg{k}$, the same must be true for any pair of consecutive segments $\seg{i'}$ and $\seg{j'}$ from the convex support, with $i < i' < j' < k$. Otherwise, there would be three segments from the convex support not in convex position. The algorithm starts from the last two segments in the convex support and moves to previous pairs as long as there is an intersection of the pair's common tangent with $\seg{k}$. That guarantees $g$ will be found and reported as the retargeting point in the outer ``if'' branch.

The second case, when $g$ is retargeting point of the second type, is treated similarly. In this case $g$ lies on $\Call{Tangent}{j, k}$ for some $j < k$ Figure~\ref{fig:tangent2}. There can be no $\seg{r}$ with $r < k$ such that it reaches below $\Call{Tangent}{j, k}$; otherwise $r$ would be a jumping destination for $g$. That implies that $\seg{j}$ is a part of the lower convex hull of the segments left from $\seg{k}$, which further means that $\seg{j}$ is the element of the convex support for the set $\{\seg{j_0-1}, \seg{j_0}, \ldots, \seg{k-1}\}$. Since $\Call{Tangent}{j, k}$ touches $\seg{k}$ from below, common tangent of each pair of consecutive segments $\seg{i'}$ and $\seg{j'}$, with $j \leq i' < j' < k$, from the convex support must intersect $\seg{k}$. Otherwise, there would be three segments from the convex support not in convex position. For the same reason the common tangent of $\seg{j}$ and the segment immediately before it in the convex support cannot intersect $\seg{k}$. The algorithm starts from the last two segments in the convex support and moves to previous pairs as long as there is an intersection of the pair's common tangent with $\seg{k}$. Finally, the algorithm reaches the rightmost pair of two consecutive segments from convex support whose common tangent does not intersect $\seg{k}$. The right segment from that pair is exactly $\seg{j}$. In that moment point $g$ is found and reported as retargeting point in the inner ``if'' branch.

The running time of algorithm is $O(n)$, since each index $k \in \{j_0 + 1, \ldots, j_0 + n\}$ is pushed on stack and popped from stack at most once, and output of $\Call{Tangent}$ and intersections can be computed in constant time. The number of retargeting points reported is, therefore, also $O(n)$.

Retargeting points reported by the algorithm come in order sorted from left to right, which is explained by following observations. Retargeting points reported in a single iteration of the outer for-loop belong to the same hyperbola segment, and segments come in left-to-right order. Retargeting points reported on the same hyperbola segment are also in the left-to-right order: inside the inner loop, $\seg{k}$ is consecutively intersected with lines such that each line is of lower slope than previous and lies beneath it under $\seg{k}$. Hence, each subsequent intersection point lies to the right of the previous one.

This algorithm finds only retargeting points from a single period of the plotted function, but all other retargeting points can be obtained by simply translating these horizontally by the integral number of periods $\mu(P)$.
\end{proof}

\subsubsection{Types of events}

Next, we list all types of events that can happen while moving the lever leftwards, and for each we show how to calculate the value of $t$, the lever position at which the event occurs. We will give a polynomial equation describing each event type, and it will be solved either for $t$ or $o_i(t)$. Once we have $o_i(t)$, it is easy to obtain $t$, as
\begin{equation}
\label{eq:inverseO}
t = \frac{o_i(t)^2 - d_i^2 - m_i^2}{2(o_i(t) - m_i)}.
\end{equation}

In the process of determining $t$, we will repeatedly encounter fixed degree polynomial equations. Solving them can be assumed to be a constant time operation, see \cite{dickenstein2005solving}.

We will also frequently use the following two two utility functions, $c_j(o)$ and $s_j(o)$.

For $o<m_j$, let $c_j(o)$ be such that $H(c_j(o)) := (c_j(o), h(c_j(o)))$ is the contact point of hyperbola $h_j$ and its tangent through the point $(o, 0)$. Given $o$, the value $c_j(o)$ can be calculated by solving the equation $h_j'(c_j(o)) = h_j(c_j(o)) / (c_j(o)-o)$, which results in
\begin{equation}
\label{eq:c}
c_j(o) = \frac{d_j^2 + m_j^2 - m_jo}{m_j - o} .
\end{equation}

Function $s_j(o)$ is defined as the slope of the tangent to hyperbola $h_j$ through the point $(o, 0)$, which is the line through points $(o, 0)$ and $H(c_j(o))$.
\begin{equation}
\label{eq:s}
s_j(o) = \frac{h_j(c_j(o))}{c_j(o)-o} = 1 / \sqrt{\left(\frac{m_j - o}{d_j}\right)^2 + 1} .
\end{equation}

These functions are used when we know that the lever for $t$ is tangent to some hyperbola segment $\seg{j}$. We then know that the lever is touching $\seg{j}$ at the point with coordinate $c_j(o(t))$, and that its slope equals $s_j(o(t))$.

First, we consider jump destination change event, a type of event which is not a state transition event. Nevertheless, it is still necessary to react to events of this kind in order to update a parameter needed for calculating events that do change state.

\begin{itemize}

\item Jump destination change event

Jump destination $j_m:=\jump(c(t))$ changes whenever $C(t)$ passes over some retargeting point. At that moment it is necessary to recalculate all future events which involve jumps, since $j_m$ is used for their calculation. Let $z$ be the next retargeting position, i.e.\ the rightmost one that lies to the left of $c(t_c)$, where $t_c$ is the current lever position. Depending on the lever state, we calculate the event position in one of the following ways.

\begin{itemize}
\item{The current state is $\stateV{i}{j}$}

The jump destination change event cannot occur before leaving this
state since $C(t)$ stands still at the ``wedge tip", so it cannot
pass over any retargeting point.

\item{The current state is $\stateK{i}{j}$}

Here the lever is tangent to $\seg{j}$, so this event can only happen if
$z > m_j$. Otherwise, the lever would have nonpositive slope when
touching the plot at $H(z)$. The equation describing this event is
\[ c_j(o_i(t)) = z ,\]
and it solves to
\[ o_i(t) = \frac{d_j^2-zm_j+m_j^2}{m_j-z} .\]

\item{The current state is $\stateY{i}{j}$}

Right endpoint of the lever slides over $\seg{j}$ and will coincide
with $H(z)$ when
\[ t = z -\mu(P)/2 .\]
\end{itemize}

Note that $j_m$ is not used in the description of the lever state, so, as already noted, jump destination change event does not change the current state.

\end{itemize}

All other events that need to be considered are state transition events.

In the following list we give all possible types of state transition events, and we show how to calculate corresponding $t$ value for each of them.

\begin{itemize}

\item $\state{i}{j}^x \rightarrow \state{i-1}{j}^x$, where $x\in\{\mathcal{Y}, \mathcal{K}, \mathcal{V}\}$

This is the event when the interval to which $t$ belongs changes from $[e_i, e_{i+1})$ to $[e_{i-1}, e_i)$, so this event happens at $e_i$.
\[ t = e_i .\]

\item $\stateY{i}{j} \rightarrow \stateY{i}{j-1}$

Here, the right endpoint of the lever slides continuously from one hyperbola segment to another,
\[ t = e_j - \mu(P) / 2 .\]

\item $\stateY{i}{j} \rightarrow \stateK{i}{j}$ and $\stateK{i}{j} \rightarrow \stateY{i}{j}$

In this event the lever changes from being a tangent to $\seg{j}$ to touching $\seg{j}$ with its right endpoint, or the other way round. The corresponding equation for this event is
\[ c_j(o_i(t)) = t + \mu(P) / 2 ,\]
which can be transformed to a cubic equation in $t$.

Since there can be at most three real solutions to that equation, it is possible that this event takes place at most three times with the same $i$ and $j$. On each occurrence of the event the lever switches between being a tangent and touching the plot with its right endpoint.

\item $\stateY{i}{j} \rightarrow \stateK{i}{j_m}$

This event happens when the lever state changes from having an endpoint on $\seg{j}$ to being a tangent to $\seg{j_m}$. The corresponding equation is

\[ s_{j_m}(o_i(t)) = \frac{h_j(t+\mu(P)/2)}{h_i(t) + \mu(P)/2} ,\]

which further transforms into a polynomial equation in $t$.

The line through $o_i(t)$ with slope $s_{j_m}(o_i(t))$ touches the
hyperbola $h_{j_m}$, but we need to be sure that it actually touches
the segment $\seg{j_m}$ of that hyperbola. It may as well be the
case that $\seg{j_m}$ is not wide enough to have a common point with
the line. More precisely, the first coordinate, $u$, of the touching
point between the line and $h_{j_m}$ must belong to the interval
$[e_{j_m}, e_{j_m+1})$. To get that coordinate, we solve the
equation
\[ h_{j_m}'(u) = s_{j_m}(o_i(t)) .\]
Having in mind that $o_i(t) < m_{j_m} < u$ must hold, we get a single solution
\[ u = m_{j_m} + \frac{d_{j_m}^2}{m_{j_m} - o_i(t)} .\]

If $u \notin [e_{j_m}, e_{j_m+1})$, we do not consider this event.

Checking if the line through $o_i(t)$ with slope $s_{j_m}(o_i(t))$ actually touches the hyperbola segment $\seg{j_m}$ will also be used in the calculation for one other event type, where we will refer to it by the name \emph{collision check}.

\item $\stateY{i}{j} \rightarrow \stateV{i}{j_m}$

The event when the lever state changes from having an endpoint on $\seg{j}$ to touching the wedge between $\seg{j_{m-1}}$ and $\seg{j_m}$ is described by

\[ \frac{h_{j_m}(e_{j_m})}{e_{j_m}-o_i(t)} = \frac{h_j(t+\mu(P)/2)}{h_i(t) + \mu(P)/2} ,\]

which again transforms into a polynomial equation in $t$.

\item $\stateK{i}{j} \rightarrow \stateV{i}{j}$

This event happens when the point in which the lever is touching $\seg{j}$ reaches $e_j$. Here, the lever is tangent to $\seg{j}$, and since it must have a positive slope, this will only happen if $e_j > m_j$. The event equation is
\[ c_j(o_i(t)) = e_j ,\]
which solves to
\[ o_i(t) = \frac{d_j^2-e_jm_j+m_j^2}{m_j-e_j} .\]

\item $\stateK{i}{j} \rightarrow \stateK{i}{j_m}$

This event happens when the lever becomes a tangent to two hyperbola
segments, $\seg{j}$ and $\seg{j_m}$ simultaneously. It can only
happen if $\seg{j_m}$ is lower than $\seg{j}$, i.e.\ $d_{j_m} <
d_j$,
\[ s_{j_m}(o_i(t)) = s_j(o_i(t)) .\]
Since $o_i(t) < m_{j_m}$ and $o_i(t) < m_j$, the only solution is
\[ o_i(t) = \frac{d_j m_{j_m} - d_{j_m} m_j}{d_j - d_{j_m}} .\]

Here we again need to apply the collision check described earlier to see if the common tangent actually touches $\seg{j_m}$. If the test fails, we do not consider this event.

\item $\stateK{i}{j} \rightarrow \stateV{i}{j_m}$

The event in which the lever touches the wedge tip at $e_{j_m}$
while being a tangent to $h_j$ is represented by the following
equation,

\[ \frac{h_{j_m}(e_{j_m})}{e_{j_m}-o_i(t)} = s_j(o_i(t)). \]

This can only happen if $h_{j_m}(e_{j_m}) < d_j$. It can be
transformed to a quadratic equation in $o_i(t)$. The two solutions
correspond to two tangents to $h_j$ from the point $(e_{j_m},
h(e_{j_m}))$. The smaller of the two solutions is where this event
happens.

\item $\stateV{i}{j} \rightarrow \stateY{i}{j-1}$
This event happens when the lever stops touching the tip of the
wedge and starts to slide its right endpoint over the hyperbola
segment on the left of the wedge,
\[ t = e_j - \mu(P)/2 .\]

\item $\stateV{i}{j} \rightarrow \stateK{i}{j-1}$

This event happens when the lever stops touching the tip of the
wedge and becomes a tangent of the hyperbola segment on the left of
the wedge. This can only happen if $e_j > m_{j-1}$,
\[ c_{j-1}(o_i(t)) = e_j ,\]
which solves to
\[ o_i(t) = \frac{d_{j-1}^2-e_jm_{j-1}+m_{j-1}^2}{m_{j-1}-e_j} .\]

\item $\stateV{i}{j} \rightarrow \stateK{i}{j_m}$

This event happens when the lever stops touching the tip of the
wedge and becomes a tangent of the hyperbola segment $\seg{j_m}$,

\[ s_{j_m}(o_i(t)) = \frac{h_j(e_j)}{e_j - o_i(t)}. \]

This can only happen if $h_j(e_j) > d_{j_m}$. From that we get a
quadratic equation in $o_i(t)$. The two solutions correspond to two
tangents to $h_{j_m}$ from the point $(e_j, h(e_j))$. The smaller of
the two solutions is where this event happens.

\item $\stateV{i}{j} \rightarrow \stateV{i}{j_m}$

This event happens when the lever touches two wedges, at points $E_{j_m}$ and $E_j$ simultaneously. The condition for that is

\[ \frac{h_j(e_j)}{e_j - o_i(t)} = \frac{h_{j_m}(e_{j_m})}{e_{j_m} - o_i(t)} ,\]

which is a linear equation in $o_i(t)$.

\end{itemize}

\subsubsection{Sequence of states}

We want to efficiently find the sequence of states through which the
lever will pass on its leftward journey, together with the positions
where the state changes happen. Let the obtained sequence be $p_1,
\St_1, p_2, \St_2, p_3, \ldots, p_r, \St_r$, where $p_1 \leq p_2
\leq \ldots \leq p_r$. Each state $\St_k$ occurs when the lever
position is exactly between $p_k$ and $p_{k+1}$, where $p_{r+1} =
p_1 + \mu(P)$. We call this \emph{the sequence of realized states}.

To calculate the sequence of realized states, we will start from a
specific lever position that has a known state. Let $p_{low}$ be any
of the values for which $h$ attains its global minimum, and let
$\seg{j_0}$ be the hyperbola segment above it. The algorithm
starts with the lever in position $t_c = t_0 = p_{low} - \mu(P)/2$.
This lever has its right endpoint on the plot at the point
$H(p_{low})$, which means that its state is $\stateY{i_0}{j_0}$,
where $i_0$ is the index of hyperbola segment over $t_0$. We note
that $p_{low}$ is also a retargeting position, so we also know our
initial jump destination.

The algorithm then iterates with the following operations in its
main loop. It first calculates all possible events that could happen
while in the current state. Among those events let $E$ be the one
with the largest $t$ that is not larger than $t_c$. It is
the event that must occur next. The algorithm sets $t_c$ to $t$, and
it either updates jump destination if $E$ is jump destination change
event, or switches to the new state if the event is state transition
event. In the latter case, position $t$ and the new state are added
to the sequence of realized states. These operations are iterated
until one full period of the plot is swept, ending with $t_c = t_0 -
\mu(P)/2$ in state $\stateY{i_0-n}{j_0-n}$. Described procedure is given in Algorithm~\ref{alg:SlidingLever}.

\begin{algorithm}[h]
\caption{Sliding Lever Algorithm}
\label{alg:SlidingLever}
\begin{algorithmic}
    \State Find $p_{low}$ and $j_0$.
    \State $t_0 \leftarrow p_{low} - \mu(P)/2$
    \State Find $i_0$.

    \State Run \Call{Retargeting Points}{} to find retargeting points and their jump destinations.

    \State $t_c \leftarrow t_0$
    \State $i \leftarrow i_0$
    \State $j \leftarrow j_0$
    \State $j_m \leftarrow$ jump destination of $p_{low}$

    \State Set the current state to $\stateY{i_0}{j_0}$.
    \State Add $t_0$ and $\stateY{i_0}{j_0}$ to the sequence.

    \While {$t_c > t_0 - \mu(P)$}
        \State Calculate all the events for the current state. Ignore jumping events if $j = j_m$.
        \State Let $E$ be the first event to happen (the one with the largest $t$ not larger than $t_c$).
        \State $t_c \leftarrow$ $t$ of the event $E$.
        \If {$E$ is jump destination change event}
            \State Update $j_m$.
        \Else
            \State Set the current state to the destination state of $E$.
            \State Add $t$ and the current state to the sequence.
            \If {$E$ is a jumping event}
                \State $j_m \leftarrow j$
            \EndIf
        \EndIf
    \EndWhile
\end{algorithmic}
\end{algorithm}

\begin{theorem}
\label{thm:LinearSequenceComplexity}
\Call{Sliding Lever Algorithm}{} finds the sequence of realized states in $O(n)$ time. The length of the produced sequence is $O(n)$.
\end{theorem}

\begin{proof}

During the execution of the algorithm we must encounter all realized states, since states can change only at events and by always choosing the first following possible event to happen, we eventually consider all realized events. Realized states are encountered in order, since $t_c$ is never increasing.

While choosing the following event we did not consider the possibility that there can be several events with the same, minimal, $t$. However, if that happens we can choose an arbitrary one to be the next event. This choice can influence the output sequence only by including or excluding some states of the length zero. Importantly, such zero-length states are irrelevant for further considerations, and no other state in the output of the algorithm is influenced by this choice.

Each event is either jump destination change event or state transition event. From Theorem~\ref{thm:RetargetingPositions} we have that there are $O(n)$ jump destination change events, and now we will show that there are $O(n)$ state transition events.

Each state transition event transitioning from some $\state{i}{j}$ state decreases either $i$, or $j$ or both. The only exception are the events $\stateK{i}{j} \rightarrow \stateY{i}{j}$ and $\stateY{i}{j} \rightarrow \stateK{i}{j}$, however those events can happen at most three times in total for the same $i$ and $j$. Note that $j_m \leq j$, but when $j_m = j$, we do not consider events involving $j_m$. Variables $i$ and $j$ start with values $i_0$ and $j_0$, and, after the loop finishes, they are decreased to $i_0 - n$ and $j_0 - n$. Hence, no more than $O(n)$ state transition events occurred, implying the linear length of the sequence of realized states.

Calculation of each state transition event takes a constant time, at each iteration
there is a constant number of potential events considered, and loop is
iterated $O(n)$ times. To find next jump destination change event,
we move through the sorted list of retargeting positions until we
find the first retargeting position not greater than $t_c$. The
total time for calculating jump destination change events, over all
iterations, is linear. Therefore, the running time of the whole
algorithm is also linear.
\end{proof}

\section{Merging the two dilations}
\label{sec:MergingTheTwoDilations}

Knowing the sequence of realized states is sufficient to determine
the exact lever slope at any position. Remember, the left lever
slope at position $t$ is the inverse of the left dilation via
$P(t)$, as shown in (\ref{eq:dilationSlopeConnection}). But, to know
the dilation via some point we need both left and right dilations
via that point (\ref{eq:dilationIsMaximum}).

Our sliding lever algorithm was initially designed only for left
dilation, but an analogous algorithm can obviously be designed for
the right dilation (or, we can perform the \emph{exact} same
algorithm for the left dilation on the mirror image of the polygon
$P$, and then transform obtained results appropriately). This
implies the concept of the right dilation lever for $t$ (as opposed
to the left dilation lever, or just lever, as we have been calling
it until now), which has negative slope and touches the plot on the
left side of $t$. We will use $+$ and $-$ in superscript denoting
relation with left and right dilation, respectively.

Let $p^+_1, \St^+_1, p^+_2, \St^+_2, p^+_3, \ldots, p^+_{r^+}, \St^+_{r^+}$ and $p^-_1, \St^-_1, p^-_2, \St^-_2, p^-_3, \ldots, p^-_{r^-}, \St^-_{r^-}$ be the sequences of realized states for left and right dilation, respectively, where both sequences $p^+$ and $p^-$ are in nondecreasing order. For simplicity, let us call them the \emph{left} and the \emph{right} sequence, respectively. States for right dilation are described by $\state{i}{j}$ notation as well, with the meaning analogous to the meaning of the notation for left dilation states. We say that the (right dilation) lever for $t$ is in the state $\state{i}{j}$, when $\seg{i}$ is the hyperbola segment above $t$, and the (right dilation) lever touches the hyperbola segment $\seg{j}$.

We now merge the two obtained sequences by overlapping them into a new sequence $p_1, \St_1, p_2, \St_2, p_3, \ldots, p_r, \St_r$. In the merged sequence, $p_1 \leq p_2 \leq \ldots \leq p_r$ is the sorted union of $\{p^+_1, p^+_2, \ldots, p^+_{r^+}\}$ and $\{p^-_1, p^-_2, \ldots, p^-_{r^-}\}$. States in the merged sequence are pairs consisting of one state from the left sequence and one state from the right sequence. Each state $S_k = (\St^+_{k^+}, \St^-_{k^-})$ from the merged sequence is such that $\St^+_{k^+}$ and $\St^-_{k^-}$ are states covering the interval between $p_k$ and $p_{k+1}$ in the left and in the right sequence, respectively.

By Theorem~\ref{thm:LinearSequenceComplexity}, both $r^+$ and $r^-$
are $O(n)$, so the length of the merged sequence is also linear in
$n$. Hence, the merged sequence can be computed in $O(n)$ time.

For each state $\St_k = (\St^+_{k^+}, \St^-_{k^-})$ there is a single expression for computing the lever slope as a function of $t$, when $p_k \leq t \leq p_{k+1}$, both for the left dilation and for the right dilation. To find minimal dilation while in that state, we want to find $t$ which maximizes the minimum of the two slopes for left and right lever. This observation readily follows from (\ref{eq:dilationIsMaximum}) and (\ref{eq:dilationSlopeConnection}), so

\begin{equation}
\label{eq:minInCombinedState}
    \min_{p_k \leq t \leq p_{k+1}} \dil_{P(t)} = \frac{1}{\max_{p_k \leq t \leq p_{k+1}} \min \{\s^+(t), \s^-(t)\}},
\end{equation}

where $\s^+(t)$ is the slope for the left dilation lever for $t$, and
$\s^-(t)$ is the slope for the right dilation lever for $t$.

Let us analyze the shape of the functions $\s^+(t)$ and $\s^-(t)$. Assume that the corresponding state to which $t$ belongs is $\St = (\St^+, \St^-)$.

Let $s^+_{\stateK{i}{j}}(t)$ be a function which maps $t$ to the slope of a lever, assuming that the lever is in $\stateK{i}{j}$ state. Analogously we define $s^+_{\stateV{i}{j}}(t)$, $s^+_{\stateY{i}{j}}(t)$, $s^-_{\stateK{i}{j}}(t)$, $s^-_{\stateV{i}{j}}(t)$ and $s^-_{\stateY{i}{j}}(t)$.

\begin{lemma}
\label{lem:sKmonotonicity} If $\St^+$ is a $\stateK{i}{j}$ state,
then $\s^+(t)$ is a monotonically nondecreasing function.
\end{lemma}
\begin{proof}
If $\St^+$ is an $\stateK{i}{j}$ state, then, from equation
(\ref{eq:s}), we have
$$
\s^+(t) = s^+_{\stateK{i}{j}}(t) = s_j(o_i(t)) = \frac{h_j(c_j(o_i(t)))}{c_j(o_i(t))-o_i(t)} =
1 / \sqrt{\left(\frac{m_j - o_i(t)}{d_j}\right)^2 + 1} .
$$
We see that function $s_j$ is monotonically increasing for parameter values less than $m_j$. In the specified state, $o_i(t) < m_j$ holds, and since $o(t)$ is monotonically nondecreasing (Lemma~\ref{lem:monotonicity}), it means that combination of $s_j$ and $o(t)$, which is $\s^+(t)$, is monotonically nondecreasing as well.
\end{proof}

\begin{lemma}
\label{lem:sVmonotonicity} If $\St^+$ is a $\stateV{i}{j}$ state,
then $\s^+(t)$ is a monotonically nondecreasing function.
\end{lemma}
\begin{proof}
If $\St^+$ is an $\stateV{i}{j}$ state, then we have
$$
\s^+(t) = s^+_{\stateV{i}{j}}(t) = \frac{h_j(e_j)}{e_j - o_i(t)}.
$$

We see that $\s^+(t)$ is monotonically increasing in terms of $o_i(t)$
when $o_i(t) < e_j$, which holds in the specified state. Since $o(t)$ is
monotonically nondecreasing (Lemma~\ref{lem:monotonicity}), it means that
$\s^+(t)$ is monotonically nondecreasing in terms of $t$ as well.
\end{proof}


Similar observations hold for the right dilation analogues: $\s^-(t)$
is monotonically decreasing if $\St^-(t)$ is $\stateK{i}{j}$ or
$\stateV{i}{j}$ state.

\begin{lemma}
\label{lem:YnoConsider}
If $\St^-$ is a $\stateY{i}{j}$ state then $\s^+(t) \leq \s^-(t)$.
\end{lemma}
\begin{proof}
From equation (\ref{eq:dilationSlopeConnection}), using the fact that $h_j(t) = h_{j+n}(t + \mu(P))$ holds because of the periodicity of the plot, we have
\begin{equation*}
\begin{aligned}
\s^-(t) &= \frac{h_j(t - \mu(P)/2)}{h_i(t) + \mu(P)/2} \\
        &= \frac{h_{j+n}(t + \mu(P)/2)}{h_i(t) + \mu(P)/2} \\
        &= s(t, t + \mu(P)/2) \\
        &\geq \min_{t_r \in [t, t + \mu(P)/2]} s(t, t_r) \\
        &= \s^+(t).
\end{aligned}
\end{equation*}
\end{proof}

Analogously, if $\St^+$ is a $\stateY{i}{j}$ state then $\s^-(t) \leq \s^+(t)$.

We need to calculate $\max_{p_k \leq t \leq p_{k+1}} \min \{\s^+(t), \s^-(t)\}$ in (\ref{eq:minInCombinedState}), which is equivalent to finding the highest point of the lower envelope of the functions $\s^+(t)$ and $\s^-(t)$, see Figure~\ref{fig:8}. This calculation depends on the types of the states $\St^+$ and $\St^-$. We analyze nine possible type combinations.

\fig{fig:8}{$\max_{p_k \leq t \leq p_{k+1}} \min \{\s^+(t), \s^-(t)\}$}{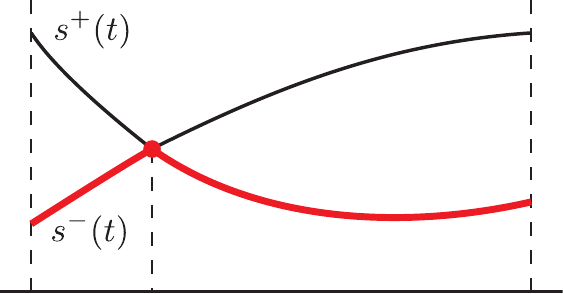}

\begin{itemize}
    \item If $\St^+$ is $\stateY{i}{j^+}$ state and $\St^-$ is $\stateY{\i}{j^-}$
    state:

    Using Lemma~\ref{lem:YnoConsider} we get $\s^+(t) = \s^-(t)$, so
    \[\max_{p_k \leq t \leq p_{k+1}} \min \{\s^+(t), \s^-(t)\} = \max_{p_k \leq t \leq p_{k+1}} s^+_{\stateY{i}{j^+}}(t).\]
    The maximum is achieved either at interval ends or at local maxima, if one exists, which is found by solving a polynomial equation.

    \item  If $\St^+$ is $\stateK{i}{j^+}$ state and $\St^-$ is $\stateY{i}{j^-}$ state:

    Using Lemma~\ref{lem:YnoConsider} and Lemma~\ref{lem:sKmonotonicity} we get
    \[\max_{p_k \leq t \leq p_{k+1}} \min \{\s^+(t), \s^-(t)\} = \max_{p_k \leq t \leq p_{k+1}} s^+_{\stateK{i}{j^+}}(t) = s^+_{\stateK{i}{j^+}}(p_{k+1}).\]

    \item  If $\St^+$ is $\stateK{i}{j^+}$ state and $\St^-$ is $\stateK{i}{j^-}$ state:

    From Lemma~\ref{lem:sKmonotonicity} we know that $s^+_{\stateK{i}{j^+}}(t)$ is monotonically nondecreasing, and $s^-_{\stateK{i}{j^-}}(t)$ is monotonically nonincreasing. The highest point of the lower envelope of their plot on $[p_k, p_{k+1}]$ is thus located either at one of the interval endpoints, or at the point of the intersection of two plots, which can be found by solving a polynomial equation.

\end{itemize}

The other six combinations of state types are not listed, but each of them is resolved in a manner similar to the one of the above three combinations. Cases with $\stateV{\cdot}{\cdot}$ are resolved analogously to cases that have $\stateK{\cdot}{\cdot}$ instead, by using Lemma~\ref{lem:sVmonotonicity} in place of Lemma~\ref{lem:sKmonotonicity}, and the remaining cases are analogous to the cases having ``pluses'' and ``minuses'' swapped.

Finally, by taking the smallest of all dilation minima from $[p_k, p_{k+1}]$ intervals for $k \in \{1, 2, \ldots r\}$ we obtain the overall minimum dilation,
\[ \delta = \min_{k \in \{1,2,\ldots,r\}} \min_{p_k \leq t \leq p_{k+1}} \dil_{P(t)}.\]

While going through calculated interval minima we maintain the value of $t$ for which the minimum is achieved, so we also get the point on $P$ which is the endpoint of the optimal feed-link.

\section{Conclusion}

The problem we considered asked for the optimal extension of polygonal network by connecting a specified point to the rest of the network via a feed-link. We gave a linear time algorithm for solving this problem, thus improving upon previously best known result of Aronov et al.\ presented in~\cite{aronov2012connect}.

On the way to solution, we performed several steps. First, we divided the concept of dilation into the left and right dilation, which can be analyzed separately. Then we transformed them into the problem which considers plot of the distance function and lever slopes. An algorithm for event based simulation of lever movement is given. The output of the algorithm is description of the changes in lever slope presented as a sequence of states, each of which can be expressed analytically. Finally, we explained how those state sequences for left and right dilation can be merged and how the optimal feed-link can be found from it.

The method we used for solving the original problem can easily be adapted to work with any network shaped as an open polygonal chain.

Aronov et al.\ in~\cite{aronov2012connect} discuss polygons with obstacles. They show how $b$ obstacles induce a partition of the polygon boundary of the size $O(nb)$. Each segment of that partition has a distance function to $p$ similar to function $h_i(t)$, the only difference being a constant added outside the square root. This makes our hyperbola segments in the plot to shift upward by that constant, so the problem with obstacles looks very similar to the one without them. It is thus reasonable to expect that our method for solving the original problem could be adapted to solve the case with obstacles in $O(nb)$ time.

One generalization of the problem is when polygon edges are not line segments, but some other curves (i.e.\ second order curves). The abstraction behind our method can be applied in this case if there is an efficient way to determine event times and to find optimal values in the merged sequence. It would also be interesting to see whether a similar method can be applied to a network which is not necessarily polygonal, that is, when some vertices can have degree greater than two.

\bibliographystyle{unsrt}
\bibliography{references}

\begin{thebibliography}{1}

\bibitem{dahlgren2006evaluation}
A.~Dahlgren and L.~Harrie.
\newblock Evaluation of computational methods for connecting points to large
  networks.
\newblock {\em Mapping and Image Science}, (4):45--54, 2006.

\bibitem{dahlgren2007development}
A.~Dahlgren and L.~Harrie.
\newblock Development of a tool for proximity applications.
\newblock In {\em Proceedings of Agile}, 2007.

\bibitem{aronov2012connect}
B.~Aronov, K.~Buchin, M.~Buchin, B.~Jansen, T.~de~Jong, M.~van Kreveld,
  M.~Loffler, J.~Luo, R.I. Silveira, and B.~Speckmann.
\newblock Connect the dot: Computing feed-links for network extension.
\newblock {\em Journal of Spatial Information Science}, (3):3--31, 2012.

\bibitem{ebbers2006geometric}
A.~Ebbers-Baumann, A.~Grune, and R.~Klein.
\newblock The geometric dilation of finite point sets.
\newblock {\em Algorithmica}, 44(2):137--149, 2006.

\bibitem{dickenstein2005solving}
A.~Dickenstein and I.Z. Emiris.
\newblock {\em Solving polynomial equations: Foundations, algorithms, and
  applications}, volume~14.
\newblock Springer Verlag, 2005.

\bibitem{andrew1979another}
A.M. Andrew.
\newblock Another efficient algorithm for convex hulls in two dimensions.
\newblock {\em Information Processing Letters}, 9(5):216--219, 1979.

\end{thebibliography}

\end{document}